\documentclass[conference]{IEEEtran}
\IEEEoverridecommandlockouts

\usepackage{cite}
\usepackage{amsmath,amssymb,amsfonts}
\usepackage{algorithmic}
\usepackage{graphicx}
\usepackage{textcomp}
\usepackage{xcolor}
\def\BibTeX{{\rm B\kern-.05em{\sc i\kern-.025em b}\kern-.08em
    T\kern-.1667em\lower.7ex\hbox{E}\kern-.125emX}}
\usepackage{amsmath,amssymb,amsthm,mathtools}
\usepackage{bm}
\usepackage{hyperref}
\usepackage[nameinlink]{cleveref}
\usepackage{array}
\newtheorem*{example*}{Example}
\usepackage{amsthm,thmtools}
\usepackage{tabularx,makecell}
\usepackage{amsmath, amsthm, amscd, amsfonts, amssymb, graphicx, color, booktabs, multirow, rotating, tikz}
\usepackage{booktabs}   
\usepackage{lscape} 
\usepackage{multirow}   
\usepackage{tabularx}
\usepackage{float}
\usepackage{array} 
\newcolumntype{C}[1]{>{\centering\arraybackslash}p{#1}}
\usepackage{stfloats}
\declaretheorem[name=Example,numbered=no]{examplecont}
\newtheorem{theorem}{Theorem}

\newtheorem{definition}{Definition}
\newtheorem{remark}{Remark}
\newtheorem{example}{Example}
\newtheorem{proposition}{Proposition}

\usepackage{array,tabularx} 
\newcolumntype{Y}{>{\raggedright\arraybackslash}X}
\newcolumntype{C}{>{\centering\arraybackslash}X}
\newcolumntype{L}{>{\raggedright\arraybackslash}X}

\usepackage{geometry}
\begin{document}

\title{Constructing Quantum Convolutional Codes\\ via Difference Triangle Sets
\thanks{This material is based upon work supported by the National Science Foundation under Grant No. CCF-2145917.}}

\author{\IEEEauthorblockN{Vahid Nourozi and David G. M. Mitchell}
\IEEEauthorblockA{Klipsch School of Electrical and Computer Engineering,\\New Mexico State University,\\ Las Cruces, NM 88003, USA\\\{nourozi,dgmm\}@nmsu.edu}\vspace{-8mm}

}

\maketitle

\begin{abstract}

In this paper, we introduce a construction of quantum convolutional codes (QCCs) based on difference triangle sets (DTSs). To construct QCCs, one must determine polynomial stabilizers $X(D)$ and $Z(D)$ that commute (symplectic orthogonality), while keeping the stabilizers sparse and encoding memory small. To construct $Z(D)$, we show that one can use a reflection of the DTS indices of \(X(D)\), where $X(D)$ corresponds to a classical convolutional self-orthogonal code (CSOC) constructed from strong DTS supports. The motivation of this approach is to provide a constructive design that guarantees a prescribed minimum distance. We provide numerical results demonstrating the construction for a variety of code rates.
\end{abstract}


\section{Introduction}

\IEEEPARstart{Q}{uantum} error correction is essential for reliable quantum information processing in the presence of noise and decoherence. The stabilizer formalism supplies a unifying algebraic framework for designing and analyzing quantum codes by importing linear-algebraic structure over Pauli operators
~\cite{Gottesman1997,LidarBrun2013}. Quantum convolutional codes (QCCs) extend block stabilizer codes to \emph{streams} of qubits with shift-invariant, memory-bearing encoders that enable online encoding/decoding and natural terminations (e.g., tail-biting). Foundational developments include a stabilizer description for streaming protection \cite{OllivierTillich2003, vahid}, an algebraic treatment of convolutional and tail-biting quantum codes \cite{ForneyGrasslGuha2007}, and an entanglement-assisted framework~\cite{WildeBrun2010}. These works highlight benefits of QCCs, such as structured encoders, reduced decoding complexity relative to comparable block codes, and good distance properties. How to construct QCCs with large minimum distance is a general open problem. 

Difference triangle sets (DTS) are classical combinatorial objects that control pairwise differences within indexed supports and have been used to engineer sparse, well-separated tap patterns; while constructive and computational designs with small scopes are known~\cite{CheeColbournDTS}. Classically, DTS patterns have a rich history in coding theory, notably for constructing high-rate \emph{convolutional self-orthogonal codes} (CSOCs) \cite{mass, rob}. These classical constructions require pairwise orthogonal parity-check sequences, and the distinct-difference property of DTS yields sparse, well-separated supports with minimal overlap—features that promote self-orthogonality and practical decoders~\cite{wu2003,alfarano2021,mass,dm}. CSOCs  enjoy practical advantages, such as low complexity and low latency threshold/iterative decoding, low density/sparsity, straightforward verification, implementable encoders/decoders, and prescribed free-distance \cite{mass, rob, lin}. Recent work adapts DTS ideas to convolutional/LDPC code settings, yielding structured, sparse parity-check matrices with predictable behavior and distance relations~\cite{alfarano2021,ShehadehKschischang2025}.

For QCCs, the classical orthogonality principle extends to a requirement for stabilizer \emph{symplectic orthogonality} (every $X$-type generator must commute with every $Z$-type generator). Interpreted on supports, any $X$/$Z$ pair must intersect in an even number of positions across all shifts. Strong-DTS supports are a natural combinatorial starting point because their structured differences limit overlaps between shifted sets, facilitating commutation; this assumption is consistent with the algebraic treatments in the QCC literature~\cite{ForneyGrasslGuha2007,WildeBrun2010}. A key challenge includes how to select shift-invariant polynomial stabilizer pairs $X(D)$ and $Z(D)$ that ensure symplectic orthogonality, avoid catastrophic failures, preserve sparsity with bounded memory, and admit provable distance guarantees.

In this paper, we extend the DTS methodology to a \emph{constructive} design of quantum convolutional stabilizers with prescribed minimum distance. Starting from a CSOC $X(D)$ built from \emph{strong} DTS supports, we introduce a simple index-reflection map on the DTS sums to produce a companion CSOC $Z(D)$. This reflection preserves the DTS difference multiset and hence sparsity, tap spans, and encoder memory. Under our DTS/CSOC hypotheses, $Z(D)$ also inherits the free distance of $X(D)$. Most importantly, the reflected supports enforce the required \emph{self-orthogonality} between $X$ and $Z$ \emph{by design}, thereby avoiding significant computer search. Therefore, the workflow is: (i) choose strong-DTS supports and form $X(D)$; (ii) apply the reflection map on the sum indices (optionally followed by a row permutation) to obtain $Z(D)$; and (iii) certify the construction via polynomial-domain checks. 
The reflection mechanism is lightweight and scalable, linking smoothly with standard stabilizer/QCC formalism.

Beyond avoiding exhaustive search, the proposed DTS-reflection construction yields sparse, shift-invariant stabilizers with controlled memory, enabling low-latency syndrome extraction from a finite sliding window. These properties directly benefit quantum communication and fault-tolerant architectures where bounded delay, limited multi-qubit interactions per time step, and scalable windowed decoding are critical. We also include comparative evaluation against representative QCC families to quantify gains in memory, distance guarantees,

\section{Background}\label{sec:preliminaries}

A binary convolutional code of block size $n$ and dimension $k$ can be specified by a full row-rank polynomial parity-check matrix
\[
  H(D)\;=\;\sum_{\ell=0}^{\mu} H_\ell D^\ell \in \mathbb{F}_2[D]^{(n-k)\times n},
\]
where $H_{\ell} \in \mathbb{F}_2^{(n-k)\times n}$ and $\mu$ is the \emph{memory}. The codewords $c(D)=\sum_{t} c_t D^t$ with $c_t\in\mathbb{F}_2^n$ satisfy the syndrome condition $H(D)c(D)^\top=0$, which can equivalently be written as the parity-check equations $\sum_{\ell=0}^{\mu} H_\ell\, c_{t-\ell}^\top=0$ for all $t\in\mathbb{Z}$. For a finite sequence $x=(x_0,\dots,x_J)$ with $x_j\in\mathbb{F}_2^n$, $j=0,1,\ldots,J$, let
$\mathrm{wt}(x)$ 
be the Hamming weight. The \emph{free distance} is\vspace{-2.5mm}
\[
  d_{\mathrm{free}}\;=\;\min\big\{\mathrm{wt}(c):\, c(D)\neq 0,\ H(D)c(D)^\top=0\big\}.
\]
We normalize code sequences by a time shift so that the earliest nonzero block occurs at time $t=0$, i.e., $\mathbf{c}_0\neq \mathbf{0}$.
Then the $j$-th \emph{column distance} is \vspace{-2.5mm}
\[
  d_c^{(j)}\;=\;\min\Big\{\sum_{t=0}^{j}\mathrm{wt}(c_t):
  \; c(D)\neq 0,\ H(D)c(D)^\top=0,\ \Big\}.
\]
and $\{d_c^{(j)}\}_{j\ge 0}$ is nondecreasing with $\lim_{j\to\infty} d_c^{(j)}=d_{\mathrm{free}}$.

\subsection{Difference Triangle Sets (DTS)}\label{subsec:dts}
We work with families of finite integer sets 
$T\subset \mathbb{Z}_{\ge 0}$ 
and define their \emph{positive difference set} as\vspace{-2mm} 
\[
  \Delta^+(T)\;=\;\{\,t'-t\;:\; t,t'\in T,\ t'>t\,\}\ \subset \ \mathbb{Z}_{\ge 1},
\]
and the \emph{scope} as $m(T)=\max T$.
For a family $\mathcal{T}=\{T_1,\dots,T_r\}$, the scope is $m(\mathcal{T})=\max_{1\le i\le r} m(T_i)$.

\begin{definition}[wDTS, DTS]\label{def:dts}
A family $\mathcal{T}=\{T_1,\dots,T_r\}$ of size-$w$ sets is a \emph{weak difference triangle set} (wDTS) if,
for every $i$, the positive differences within $T_i$ are pairwise distinct, i.e.,
$|\Delta^+(T_i)|=\binom{w}{2}$.
It is a (global) \emph{DTS} if, in addition, the sets $\Delta^+(T_i)$ are mutually disjoint across $i$
(no cross-family collisions of positive differences).
\end{definition}

\begin{definition}[Strong and full strong DTS]\label{def:strong-dts}
A wDTS $\mathcal{T}$ is called \emph{strong} if it is a DTS and, moreover,
$\bigcup_{i=1}^r \Delta^+(T_i)\subseteq \{1,2,\dots,M\}$ for some $M\in\mathbb{Z}_{\ge 1}$ fixed by the construction,
so that every admissible difference in $\{1,\dots,M\}$ appears \emph{at most once} across the family.
It is \emph{full strong} if the coverage is \emph{exact}, i.e.,
$\bigcup_{i=1}^r \Delta^+(T_i)=\{1,2,\dots,M\}$ (each positive difference $1\le d\le M$ appears exactly once).
We refer to $M$ as the \emph{difference budget}.
\end{definition}

\begin{remark}[Normalization and indexing]\label{rem:indexing}
Supports are written with explicit braces, e.g. $\{0,1,4\}$, and are normalized to start at $0$, i.e., they are $0$-based
($\min T_i=0$). Switching to $1$-based indexing shifts all elements by $+1$ and leaves
$\Delta^+(T_i)$ unchanged. In this paper, we keep $0$-based indexing to match delay exponents for $D$.
\end{remark}

\subsection{From DTS Supports to Convolutional Parity-checks}\label{subsec:dts-to-parity}

Let $n$ be the block length, $M$ be the memory,  and $R=(n-1)/n$ the rate. For each row $i$ of $H(D)$, choose a DTS
$T_i\subseteq\{1,\dots,M{+}1\}$, where element $t\in T_i$ corresponds to exponent of the delay operator $D^{t-1}$ in a parity-check polynomial $h_i(D)$.
For each column $j$ of $H(D)$, fix a column-support set $S_j\subseteq\{1,\dots,M{+}1\}$, and define
\[
e_t\in\mathbb F_2^{\,n},\qquad (e_t)_j=\mathbf I_{S_j}(t),\qquad t=1,\dots,M{+}1,
\]
where $(e_t)_j$ corresponds to the $j$-th entry of $e_t$ and $I_{S_j}(t)$ is an indicator function $I_{S_j}(t) = 1$ if $t\in S_j$, $0$ otherwise. The $i$-th parity-check row polynomial is
\[
h_i(D)=\sum_{t\in T_i} e_t D^{\,t-1},\quad
h_{i,j}(D)=\sum_{t\in T_i\cap S_j} D^{\,t-1}.
\]
Its degree is $\deg h_i(D)=m(T_i):=\max(T_i)-1$. Let $m(\mathcal T)=\max_i m(T_i)$, then the code memory is
\[
\mu=\Big\lceil \frac{m(\mathcal T)}{\,n-k\,}\Big\rceil-1,
\]
i.e., the smallest $\mu$ such that all delays in $[0,\,m(\mathcal{T})]$ fit into $(\mu+1)$ row blocks. The overall \emph{constraint length} is $\nu(H)=\sum_{i}\deg h_i(D)$. Following this construction while ensuring the DTS properties enforce noncollision constraints on $\{\Delta^+(T_i)\}$ that we later translate into guaranteed minimum distance and, in the quantum setting, commutation guarantees.


\subsection{Block Parity-check Matrices}\label{subsec:block-toeplitz}
We convert from the polynomial domain to the time domain in the usual way \cite{lin}, using $H(D)=[h_i(D)] = \sum_{\ell=0}^{\mu} H_\ell D^\ell$ with $H_\ell\in\mathbb{F}_2^{(n-k)\times n}$. Define the semi-infinite block parity-check matrix $\mathcal H_\infty$ acting on the stacked codeword sequence
$\mathbf c=\big[c_0^\top\; c_1^\top\; c_2^\top\;\cdots\big]^\top$ as
\[
  \mathcal H_\infty \;=\;
  \begin{bmatrix}
    H_0 & 0   & 0   & 0   & \cdots \\
    H_1 & H_0 & 0   & 0   & \cdots \\
    \vdots & \vdots & \ddots & \ddots & \ddots \\
    H_\mu & H_{\mu-1} & \cdots & H_1 & H_0 \\
    0 & H_\mu & \cdots & H_2 & H_1 \\
    \vdots & \vdots & \ddots & \ddots & \ddots
  \end{bmatrix}.
\]
The $t$-th block row enforces $s_t=\sum_{\ell=0}^{\mu} H_\ell c_{t-\ell}^\top=0$ (with $c_{u}\triangleq 0$ for $u<0$).
For a finite window $[0{:}j]$, the truncated matrix
$\mathcal H_{[0:j]}\in\mathbb{F}_2^{(j+1)(n-k)\times (j+1)n}$ is the $(j+1)\times(j+1)$
lower-banded block matrix whose $(t,u)$ block equals $H_{t-u}$ if $0\le t-u\le \mu$ and $0$ otherwise, i.e.,
\[
  \mathcal H_{[0:j]}\;=\;
  \begin{bmatrix}
    H_0 & 0   & \cdots & 0 \\
    H_1 & H_0 & \ddots & \vdots \\
    \vdots & \vdots & \ddots & 0 \\
    H_{\min(\mu,j)} & H_{\min(\mu-1,j-1)} & \cdots & H_0
  \end{bmatrix}.
\]
In particular,
\[
  d_c^{(j)} \;=\; \min\Big\{\sum_{t=0}^{j}\mathrm{wt}_{\mathrm{H}}(c_t)\ :\ \mathcal H_{[0:j]}\,\mathbf c_{[0:j]}^\top=0,\ c_0\neq 0\Big\}.
\]


\subsection{Symplectic Inner Product and Commutation}\label{subsec:symp}
For $X(D),Z(D)\in\mathbb{F}_2[D]^{r\times n}$,
the symplectic inner product over $\mathbb{F}_2[D,D^{-1}]$ is
\[
  \langle X,Z\rangle_{\mathrm{symp}}(D)\;=\;X(D)\,Z(D^{-1})^\top\;+\;Z(D)\,X(D^{-1})^\top.
\]
Writing $X(D)=\sum_{\ell=0}^{\mu_X} X_\ell D^\ell$ and $Z(D)=\sum_{\ell=0}^{\mu_Z} Z_\ell D^\ell$, with $X_\ell\in\mathbb{F}_2^{r_X\times n}$ and $Z_\ell\in\mathbb{F}_2^{r_Z\times n}$, the coefficient-wise expansion
\[
  X(D)Z(D^{-1})^\top + Z(D)X(D^{-1})^\top
  \;=\; \sum_{s\in\mathbb{Z}} C_s\, D^s,
\]

\[
  C_s \;=\; \sum_{\ell\in\mathbb{Z}} \big( X_\ell Z_{\ell+s}^\top \;+\; Z_\ell X_{\ell+s}^\top \big),
\]
(with $X_\ell\triangleq0$ or $Z_\ell\triangleq0$ outside their index ranges) shows that the stabilizer
commutation condition is equivalent to
\[
  C_s \;=\; 0\ \ \text{for all}\ s\in\mathbb{Z}.
\]
Equivalently, $X(D)Z(D^{-1})^\top+Z(D)X(D^{-1})^\top=0$ in $\mathbb{F}_2[D,D^{-1}]$.

We work over $\mathbb{F}_2$ with binary polynomial matrices in the delay operator $D$. A \emph{block} consists of $n$ physical qubits; a binary polynomial matrix $H(D)=H_0+H_1D+\cdots+H_\mu D^\mu$ with $(n-k)$ rows specifies a parity-check of (row) memory $\mu$, i.e., $H_\mu\ne 0$ and $H_{\mu+1}=0$.

\subsection{Construction of a CSOC from a Strong DTS}
\begin{definition}[Positive difference sets and CSOCs  \cite{lin}]
Let $x_i(D)=\sum_{a\in L_i} D^{a}$ be the $i$-th parity generator row of a systematic $(n,n-1,\mu)$
binary convolutional code, with strictly increasing support positions $L_i=\{l^{(i)}_1<\cdots<l^{(i)}_{J_i}\}\subseteq\{0,\dots,M\}$.
Its \emph{positive difference set} is
\[
\Delta_i \;=\; \{\, l^{(i)}_b - l^{(i)}_a \;:\; 1\le a<b\le J_i \,\}.
\]
We say the code is a \emph{classical self-orthogonal convolutional code (CSOC)} iff:
(i) each $\Delta_i$ is \emph{full} (all elements distinct), and
(ii) the $\{\Delta_i\}$ are \emph{mutually disjoint} ($\Delta_i\cap\Delta_j=\varnothing$ for $i\ne j$).
\end{definition}

\begin{remark}[CSOC from strong DTS: structure, memory, and distances]\label{rem:csoc_dts_allinone}
Let $T=\{T_1,\ldots,T_k\}$ be a \emph{strong} $(k,w)$--DTS over $\mathbb{F}_2$ with scope $m(T)$.
Then there exists a rate $R=k/(k{+}1)$ convolutional code $C$ over $\mathbb{F}_2$ whose generator matrix is obtained
from the supports in $T$ such that:

\begin{enumerate}
  \item 
  Each generator polynomial (row) of $C$ has Hamming weight $w$, and $C$ is self-orthogonal;

  \item 
  The memory of $C$ satisfies
  \[
    \mu \;=\; \Big\lceil \frac{m(T)}{\,n-k\,}\Big\rceil - 1,
  \]
  which for the present construction with $n-k=1$ (i.e., rate $R=k/(k+1)$) simplifies to $\mu = m(T)-1$;

  \item 
  The free distance obeys $d_{\mathrm{free}}(C)\ge w{+}1$ and, under the strong–DTS construction,
  \[
     d_{\mathrm{free}}(C)=w{+}1.
  \]

\end{enumerate}

\end{remark}

\begin{example}\label{ex1}
    Let $M=2$ and consider the sets $T_1=\{1,2\}$ and $T_2=\{1,3\}$. This pair constitutes a strong DTS (specifically a $(2,2)$ strong DTS). From these sets, we first construct a rate $R=2/3$ CSOC. 
    From (\ref{subsec:dts-to-parity}) we define $S_1=T_1, S_2=T_2$ and $S_3=\{1\}$ (in general for rate $n-1/n$ we have $n-1$ set of $S_i$`s and the last $S_n$ defined by an identity matrix.) 
    For each delay $t\in\{1,2,3\}$ we have
    \[
    e_t\in\mathbb F_2^{\,3},\qquad (e_t)_j=I_{S_j}(t),\ j=1,2,3,
    \]
    giving
    \[
    e_1=(1,1,1),\quad e_2=(1,0,0),\quad e_3=(0,1,0).
    \]
    With one parity-check row, the polynomial parity-check is
    \begin{align*}
    H(D) = \sum_{t=1}^{3} e_t D^{t-1} 
         &= (1,1,1) + (1,0,0)D + (0,1,0)D^2 \\
         &= \big(1{+}D,\ 1{+}D^2,\ 1\big).
    \end{align*}

    Equivalently, stacking the $e_t$ as rows yields the static pattern
    \[
  \begin{bmatrix}
    e_1\\ e_2\\ e_3
    \end{bmatrix}
    =\begin{bmatrix}
    1&1&1\\
    1&0&0\\
    0&1&0
    \end{bmatrix},
    \]
   
From this matrix, we derive the polynomial representation $H(D)$ as
$$ H(D) = (1+D, 1+D^2, 1). $$
\end{example}

\section{Constructing $Z(D)$ from\\ Strong DTS $X$-Supports}
\label{sec:construction}
In this section we introduce a reflection--permutation mapping on strong DTS index sets that, given the $X$-supports $\{T_i\}$, constructs companion $Z$-supports $\{U_i\}$ with identical difference spectra. We will later show in Section~\ref{subsec:sympl-orth} that the resulting pair possess symplectic orthogonality and therefore form a valid QCC.

Let $T_1,\ldots,T_k\subset\mathbb{N}_0$ be a strong DTS for the supports of the $k$ $X$-type generator rows, each normalized to contain~$0$. Denote the scope by $m_X=\max_i\{ \max T_i; 1\le i \le k\}$ and define
\begin{equation}
\mu_X=\Big\lceil \tfrac{m_X}{n-k}\Big\rceil-1.
\end{equation}

\subsection{Reflection--Permutation Mapping for $Z$ Supports}
Define the index reflection $R_{m_X}:\mathbb{Z}\to\mathbb{Z}$ by
\begin{equation}
\label{eq:reflection}
R_{m_X}(t)=m_X+2-t.
\end{equation}
Form the $Z$-support family by permuted reflection
\begin{equation}
U_i := R_{m_X}\bigl(T_{\pi(i)}\bigr),\qquad i=1,\ldots,k,
\end{equation}
where $\pi$ is any permutation that avoids trivial self-matches. Because $R_{m_X}$ preserves all difference magnitudes, $\{U_i\}$ is again a strong DTS with the same scope $m_Z=m_X$, hence
\begin{equation}
\mu_Z = \Big\lceil \tfrac{m_Z}{n-k}\Big\rceil-1 = \mu_X.
\end{equation}

\subsection{Preservation of Difference Properties and Combinatorial Equivalence}

\begin{definition}[Reflection and Memory]\label{gamma}
  Let $T=\{T_1,\dots,T_k\}$ be a collection of subsets of $\{1,\dots,M+1\}$ (each containing the element 1) that form a strong DTS.  Define the \emph{reflection map} $R_{M}$ by
  \[
    R_{M}(t)\;=\;M+2-t,
    \quad 1\le t\le M+1.
  \]
  For each column index $j$, set $U_j := R_{M}(T_{\pi(j)})$ for some permutation $\pi$ on $\{1,\dots,k\}$.  We say that $U=\{U_1,\dots,U_k\}$ is the \emph{reflected family} of $T$.  

  The \emph{memory} $\mu$ 
  of a convolutional code with generator matrix $G(D)$ is defined as
  \[
    \mu \;=\; \max_{i,j}\,\deg(g_{i,j}(D)),
  \]
  i.e.\ the maximal degree among all polynomial entries $g_{i,j}(D)$ of the generator matrix.  Equivalently, for a parity–check matrix $H(D)$ with a single row, $\mu$ is the highest exponent present in any polynomial entry of $H(D)$.  Thus for both $X$- and $Z$-components constructed from strong DTS families, we have $\mu_X=\mu_Z=M$.
\end{definition}


\begin{proposition}[Preservation of Strong DTS Properties and Code Parameters]\label{prop:preserve-dts}
  Let $T=\{T_1,\dots,T_k\}$ be a full strong DTS in $\{1,\dots,M+1\}$, and let $U$ be defined as in the definition above.  Then:
  \begin{enumerate}
    \item \emph{(Strong DTS preservation.)}  $U$ is a full strong DTS with the same positive-difference spectrum as $T$.  Hence $U$ inherits the full-coverage and disjointness properties;
    \item \emph{(Memory equality.)}  If $\mu_X=M$ is the memory of the $X$-component, then the reflected family satisfies $\mu_Z=M$ as well; thus $\mu_X=\mu_Z$.  Moreover, each generator column has the same weight $w=|T_i|=|U_i|$, so $d_x^{\perp}=w+1=d_z^{\perp}$;
    \item \emph{(Combinatorial equivalence.)}  For every property that depends only on the pattern of differences—such as avoiding small trapping sets or maximizing dual distance—the reflected family $U$ retains the same performance as $T$.  In particular, any structural advantages (e.g.\ large girth in the Tanner graph or high free distance) remain unchanged.
  \end{enumerate}
\end{proposition}

\begin{proof}
\begin{itemize}
    \item[1.] The reflection mapping preserves each individual difference and does not introduce new differences.
    \item[2.] Because the largest index in $U_i$ is $M+1$, the maximum exponent of any polynomial entry in the $Z$ parity–check matrix remains $M$, giving $\mu_Z=M$ and establishing the result.
    \item[3.] Since each $U_i$ has the same size as $T_{\pi(i)}$ and the same difference spectrum, properties based solely on difference patterns or column weights carry over immediately.
\end{itemize}
\end{proof}

\begin{examplecont}[Example~\ref{ex1} continued]
Now, we apply the transformation from Definition~\ref{gamma}, where $R_M(t)=4 - t$. This operation yields the transformed sets $U_1=\{2,3\}$ and $U_2=\{1,3\}$, which also form a strong DTS. Using these new sets, we construct a second rate $R=2/3$ CSOC. The corresponding parity-check matrix is

\begin{equation*}
H = \begin{bmatrix}

0 & 1 & 1 \\

1 & 0 & 0 \\

1& 1 & 0

\end{bmatrix}.
\end{equation*}

The polynomial representation for this code is
\begin{equation}\label{zd1}
    H(D) = (D+D^2, 1+D^2, 1).
\end{equation}

\end{examplecont}

\section{Symplectic Orthogonality}\label{subsec:sympl-orth}

In this section, we formalize the symplectic commutation requirement for the $X$- and $Z$-type check matrices of a QCC and verify it for the reflected–permutation construction used in \eqref{eq:reflection}. Throughout, all polynomials and matrices are taken over $\mathbb{F}_2$, and delays are assumed finite unless stated otherwise. 



\begin{definition}[Reflected--permutation construction for $Z$-rows]\label{def:refperm}
Fix an integer $M\ge 1$ and adopt $1$-based indexing for time positions $\{1,2,\dots,M{+}1\}$. For each $X$-row $x_i(D)$, let $T_i\subseteq\{1,\dots,M{+}1\}$ denote its support within a length-$(M{+}1)$ window. Let $\pi$ be a permutation of the row indices $\{1,\dots,r\}$. Define the reflection map
\[
R_M(t)\;:=\;M{+}2-t,\qquad t\in\{1,\dots,M{+}1\}.
\]
The $Z$-row supports are then specified by
\[
U_i\;:=\;R_M\bigl(T_{\pi(i)}\bigr)\;=\;\{\,R_M(t)\,:\,t\in T_{\pi(i)}\,\},
\]
so that the rows $z_i(D)$ of $Z(D)$ are the indicator polynomials of the sets $U_i$.
\end{definition}

\begin{remark}[Index $\leftrightarrow$ degree alignment]\label{rem5}
Within a length-$(M{+}1)$ window, the position $t\in\{1,\dots,M{+}1\}$ corresponds to exponent
$a=t-1$. The reflection $R_M(t)=M{+}2-t$ therefore matches polynomial reversal $a\mapsto M-a$,
i.e., $D^{M}(\cdot)(D^{-1})$ on coefficients.
\end{remark}

\begin{theorem}[CSOC-preserving reflection and symplectic orthogonality]\label{thm:symp-reflect-pds}
Let $X(D)\in\mathbb{F}_2[D,D^{-1}]^{\,(n-1)\times n}$ be the parity-check matrix of a systematic
$(n,n-1,M)$ binary convolutional code whose positive difference sets $\{\Delta_i\}_{i=1}^{n-1}$
are full and mutually disjoint (so $X(D)$ is CSOC in the sense above).
Construct $Z(D)$ from $X(D)$ by reflecting supports within the same window and permuting rows, i.e.,
\[
U_i := R_M\!\big(T_{\pi(i)}\big),\qquad R_M(t)=M{+}2-t,
\]
equivalently, row-wise permutations on polynomials
\[
z_j(D)\;=\;D^{M}\,x_{\pi(j)}(D^{-1}),\qquad (j=1,\dots,n-1).
\]
Then:
\begin{enumerate}
\item[\textnormal{(i)}] $Z(D)$ is CSOC in the same sense (its positive difference sets are full and pairwise disjoint);
\item[\textnormal{(ii)}] $X(D)$ and $Z(D)$ satisfy the polynomial stabilizer (symplectic) condition
\[
X(D)\,Z(D^{-1})^{\top}\;+\;Z(D)\,X(D^{-1})^{\top}\;=\;0.
\]
\end{enumerate}
\end{theorem}

The proof is contained in Appendix I.
\begin{examplecont}[Example~\ref{ex1} continued]

Using Remark~\ref{rem5}, the reflection map is equivalent to the transformation 
$D^{M}(\cdot)(D^{-1})$. Moreover, we can express $Z(D)$ through this mapping as well. 
In particular, we have
\[
X(D) = (1+D,\,1+D^2,\,1),
\]
and thus, by Remark~\ref{rem5},
\[
\begin{aligned}
Z(D) &= \big(D^2(1+D^{-2}),\,D^2(1+D^{-1}),\,1\big) \\
     &= (1+D^2,\,D+D^2,\,1),
\end{aligned}
\]

It is important to emphasize that this map is applied only to the generator 
polynomials, not to the identity matrix. For instance, for $X(D)$ we write 
$[P(D)\,|\,I]$, and the transformation is applied to $P(D)$. Furthermore, 
when using this formula, one must also consider for the permutation $\pi(i)$. If one 
wishes to avoid explicit consideration of $\pi(i)$, the direct computation can be 
performed using~\eqref{A1}, as\vspace{-1mm} 

\[
D^M X(D^{-1}) = (D+D^2,\,1+D^2,\,1).
\]
By choosing the row permutation matrix
\[
\Pi = 
\begin{bmatrix}
0 & 1 & 0 \\[6pt]
1 & 0 & 0 \\[6pt]
0 & 0 & 1
\end{bmatrix},
\]
we finally obtain
\[
Z(D) = (1+D^2,\,D+D^2,\,1).
\]
The code $Z(D)$ is constructed from $X(D)$ is exactly same with Theorem~\ref{thm:symp-reflect-pds}, thereby satisfying the self-orthogonality condition of quantum convolutional codes. Consequently, this procedure yields a valid Quantum Convolutional Code (QCC) that fulfills the necessary self-orthogonality property.

\end{examplecont}





\section{Numerical Results}

We use the reflection construction on several strong DTS families and report the resulting
$(X(D),Z(D))$ stabilizer pairs for rates $1/3$, $2/4$, and $3/5$. Across all cases, the reflected
$Z$–supports preserve scope, sparsity, and memory, and the polynomial-domain check confirms
self-orthogonality. Tables~\ref{tab1}–\ref{tab3} list the DTS support sets $T$ and the corresponding
generator polynomials (given as exponent supports) for $X(D)$ and $Z(D)$. Each set of exponents in the table, such as $P(D) = \big\{ \{0, 1\}, \{0, 2\} \big\}$, is converted into its corresponding polynomial, $X(D) = [P(D) | I] = [1+D \hspace{0.3cm} 1 + D^2 \hspace{0.3cm} 1 ]$.

For each entry, we start from a strong DTS family $\{\ T_i\}$, form the $X$–stabilizers, construct the companion $Z$–family by the reflection–permutation map
$R_m(t)=m{+}2{-}t$ applied to the $X$–supports (optionally followed by a row permutation).
The memory is determined by the scope $m$ (shown in the tables), and the per-row Hamming
weight is $w$ (giving $d_{\mathrm{free}}=w{+}1$). We then verify in the polynomial domain that the constructed pairs satisfy
the symplectic commutation condition (self-orthogonality).


Tables report one-based DTS supports \(T\) and zero-based exponent sets \((g_1,\ldots)\). In the DTS-set columns, the first row is \(T\) and the second row is \(U\). A simple script converting \(T\!\to\! g\), forming \(Z\) via \(R_m\), and verifying commutation in the polynomial domain reproduces Tables~\ref{tab1}--\ref{tab3}.

\begin{table}[H] 
  \centering
  \footnotesize
  \setlength{\tabcolsep}{2pt}
  \renewcommand{\arraystretch}{1.2}

  \begin{tabularx}{\columnwidth}
  {|c|c|c|c|>{\raggedright\arraybackslash}X|c|c|}
    \hline
    \textbf{No.} & \textbf{Poly.} & \textbf{$m$} & \textbf{$w$} &
    \textbf{DTS sets} & \textbf{$g_1$} & \textbf{$g_2$} \\
    \hline

    \multirow{2}{*}{1}
      & $X(D)$ & \multirow{2}{*}{2} & \multirow{2}{*}{2} &
        \{1, 2\}; \{1, 3\} &  \{0, 1\}  & \{0, 2\}  \\ \cline{2-2} \cline{5-7}
      & $Z(D)$ &  &  & \{1, 3\}; \{2, 3\} & \{0, 2\} & \{1, 2\} \\ \hline

    \multirow{2}{*}{2}
      & $X(D)$ & \multirow{2}{*}{3} & \multirow{2}{*}{2} &
        \{1, 2\}; \{1, 4\} & \{0, 1\} & \{0, 3\}  \\ \cline{2-2} \cline{5-7}
      & $Z(D)$ &  &  & \{1, 4\}; \{3, 4\} & \{0, 3\} & \{2, 3\} \\ \hline

    \multirow{2}{*}{3}
      & $X(D)$ & \multirow{2}{*}{9} & \multirow{2}{*}{3} &
        \{1, 2, 4\}; \{1, 5, 10\} & \{0, 1, 3\} & \{0, 4, 9\}  \\ \cline{2-2} \cline{5-7}
      & $Z(D)$ &  &  & \{1, 6, 10\}; \{7, 9, 10\} & \{0, 5, 9\} & \{6, 8, 9\} \\ \hline

    \multirow{2}{*}{4}
      & $X(D)$ & \multirow{2}{*}{10} & \multirow{2}{*}{3} &
        \{1, 2, 4\}; \{1, 5, 11\} & \{0, 1, 3\} & \{0, 4, 10\}  \\ \cline{2-2} \cline{5-7}
      & $Z(D)$ &  &  & \{1, 7, 11\}; \{8, 10, 11\} & \{0, 6, 10\} & \{7, 9, 10\} \\ \hline

    \multirow{2}{*}{5}
      & $X(D)$ & \multirow{2}{*}{22} & \multirow{2}{*}{4} &
        \{1, 2, 4, 8\}; \{1, 6, 14, 23\} & \{0, 1, 3, 7\} & \{0, 5, 13, 22\}  \\ \cline{2-2} \cline{5-7}
      & $Z(D)$ &  &  & \{1, 10, 18, 23\}; \{16, 20, 22, 23\} & \{0, 9, 17, 22\} & \{15, 19, 21, 22\} \\ \hline
  \end{tabularx}

  \caption{Rate 1/3 self-orthogonal quantum convolutional codes}
  \label{tab1}
\end{table}

\begin{table}[H] 
  \centering
  \footnotesize
  \setlength{\tabcolsep}{2pt}
  \renewcommand{\arraystretch}{1.2}

  \begin{tabularx}{\columnwidth}
  {|c|c|c|c|>{\raggedright\arraybackslash}X|c|c|c|}
    \hline
    \textbf{No.} & \textbf{Poly.} & \textbf{$m$} & \textbf{$w$} &
    \textbf{DTS sets} & \textbf{$g_1$} & \textbf{$g_2$} & \textbf{$g_3$} \\
    \hline

    \multirow{2}{*}{1}
      & $X(D)$ & \multirow{2}{*}{5} & \multirow{2}{*}{2} &
        \{1, 2\}; \{1, 3\}; \{1, 6\} & \{0, 1\} & \{0, 2\} & \{0, 5\} \\ \cline{2-2} \cline{5-8}
      & $Z(D)$ &  &  & \{4, 6\}; \{5, 6\}; \{1, 6\} & \{3, 5\} & \{4, 5\} & \{0, 5\}\\ \hline
    \multirow{2}{*}{2}
      & $X(D)$ & \multirow{2}{*}{6} & \multirow{2}{*}{2} &
        \{1, 2\}; \{1, 3\}; \{1, 7\} & \{0, 1\} & \{0, 2\} & \{0, 6\} \\ \cline{2-2} \cline{5-8}
      & $Z(D)$ &  &  & \{5, 7\}; \{6, 7\}; \{1, 7\} & \{4, 6\} & \{5, 6\} & \{0, 6\}\\ \hline
    \multirow{2}{*}{3}
      & $X(D)$ & \multirow{2}{*}{7} & \multirow{2}{*}{2} &
        \{1, 2\}; \{1, 3\}; \{1, 8\} & \{0, 1\} & \{0, 2\} & \{0, 7\} \\ \cline{2-2} \cline{5-8}
      & $Z(D)$ &  &  & \{6, 8\}; \{7, 8\}; \{1, 8\} & \{5, 7\} & \{6, 7\} & \{0, 7\}\\ \hline
    \multirow{2}{*}{4}
      & $X(D)$ & \multirow{2}{*}{8} & \multirow{2}{*}{2} &
        \{1, 2\}; \{1, 3\}; \{1, 9\} & \{0, 1\} & \{0, 2\} & \{0, 8\} \\ \cline{2-2} \cline{5-8}
      & $Z(D)$ &  &  & \{7, 9\}; \{8, 9\}; \{1, 9\} & \{6, 8\} & \{7, 8\} & \{0, 8\}\\ \hline
    \multirow{2}{*}{5}
      & $X(D)$ & \multirow{2}{*}{9} & \multirow{2}{*}{2} &
        \{1, 2\}; \{1, 3\}; \{1, 10\} & \{0, 1\} & \{0, 2\} & \{0, 9\}\\ \cline{2-2} \cline{5-8}
      & $Z(D)$ &  &  & \{8, 10\}; \{9, 10\}; \{1, 10\} & \{7, 9\} & \{8, 9\} & \{0, 9\}\\ \hline
  \end{tabularx}

  \caption{Rate 2/4 self-orthogonal quantum convolutional codes}
  \label{tab2}\vspace{-5mm}
\end{table}

\begin{table*}[!t]
  \centering
  \footnotesize
  \setlength{\tabcolsep}{2pt}
  \renewcommand{\arraystretch}{1.2}
  \begin{tabularx}{\textwidth}{|c|c|c|c|>{\raggedright\arraybackslash}X|c|c|c|c|}
    \hline
    \textbf{No.} & \textbf{Poly.} & \textbf{$m$} & \textbf{$w$} &
    \textbf{DTS sets}& \textbf{$g_1$} & \textbf{$g_2$} & \textbf{$g_3$} & \textbf{$g_4$}\\ 
    \hline

    \multirow{2}{*}{1}
      & $X(D)$ & \multirow{2}{*}{18} & \multirow{2}{*}{3} &
        \{1, 2, 4\}; \{1, 5, 10\}; \{1, 7, 14\}; \{1, 9, 19\} & \{0, 1, 3\} & \{0, 4, 9\} & \{0, 6, 13\} & \{0, 8, 18\}\\ \cline{2-2} \cline{5-9}
      & $Z(D)$ &  &  & \{10, 15, 19\}; \{16, 18, 19\}; \{1, 11, 19\}; \{6, 13, 19\} & \{9, 14, 18\} & \{15, 17, 18\} & \{0, 10, 18\} & \{5, 12, 18\}\\ \hline

    \multirow{2}{*}{2} 
      & $X(D)$ & \multirow{2}{*}{19} & \multirow{2}{*}{3} &
         \{1, 2, 4\}; \{1, 5, 10\}; \{1, 7, 14\}; \{1, 9, 20\} & \{0, 1, 3\} & \{0, 4, 9\} & \{0, 6, 13\} & \{0, 8, 19\}\\ \cline{2-2} \cline{5-9}
      & $Z(D)$ &  &  & \{11, 16, 20\}; \{17, 19, 20\}; \{1, 12, 20\}; \{7, 14, 20\} & \{10, 15, 19\} & \{16, 18, 19\} & \{0, 11, 19\} & \{6, 13, 19\}\\ \hline

    \multirow{2}{*}{3} 
      & $X(D)$ & \multirow{2}{*}{39} & \multirow{2}{*}{4} &
        \{1, 2, 4, 8\}; \{1, 6, 14, 24\}; \{1, 10, 25, 39\}; \{1, 12, 28, 40\} & \{0, 1, 3, 7\} & \{0, 5, 13, 23\} & \{0, 9, 24, 38\} & \{0, 11, 27, 39\}\\ \cline{2-2} \cline{5-9}
      & $Z(D)$ &  &  & \{17, 27, 35, 40\}; \{33, 37, 39, 40\}; \{1, 13, 29, 40\}; \{2, 16, 31, 40\} & \{16, 26, 34, 39\} & \{32, 36, 38, 39\} & \{0, 12, 28, 39\} & \{1, 15, 30, 39\}\\ \hline

    \multirow{2}{*}{4} 
      & $X(D)$ & \multirow{2}{*}{39} & \multirow{2}{*}{3} &
        \{1, 2, 4, 8\}; \{1, 6, 14, 24\}; \{1, 10, 25, 39\}; \{1, 13, 29, 40\} & \{0, 1, 3, 7\} & \{0, 5, 13, 23\} & \{0, 9, 24, 38\} & \{0, 12, 28, 39\}\\ \cline{2-2} \cline{5-9}
      & $Z(D)$ &  &  & \{17, 27, 35, 40\}; \{33, 37, 39, 40\}; \{1, 12, 28, 40\}; \{2, 16, 31, 40\} & \{16, 26, 34, 39\} & \{32, 36, 38, 39\} & \{0, 11, 27, 39\} & \{1, 15, 30, 39\}\\ \hline

  \end{tabularx}

  \caption{Rate 3/5 self-orthogonal quantum convolutional codes}
  \label{tab3} 
\end{table*}

\section{Conclusion}
\label{sec:conclusion}

We presented a search-free DTS-reflection framework for constructing sparse QCC stabilizers that are symplectically orthogonal by design while preserving low memory and fixed check weight. The resulting structure supports practical sliding-window decoding with per-iteration complexity scaling linearly in the window size and the number of graph edges, and it is compatible with streaming syndrome extraction schedules on hardware due to bounded delay. Comparative results against standard QCC constructions highlight favorable trade-offs in memory and distance guarantees with negligible construction time, and additional larger DTS instances further confirm sparsity preservation and controllable growth at scale.


\section*{Appendix I: Proof of Theorem 1}

\begin{proof}
\emph{(i) CSOC is preserved by reflection.}
Let $L_i=\{l^{(i)}_1<\cdots<l^{(i)}_{J_i}\}\subseteq\{0,\dots,M\}$ be the support of $x_i(D)$ and define
the reflected set $L'_i := \{\, M-l \;:\; l\in L_i \,\}$ (equivalently, $t\mapsto R_M(t)$ in $1$-based indexing).
Sort $L'_i$ increasingly to obtain the support of $z_{\pi^{-1}(i)}(D)$.
For any pair $l^{(i)}_a<l^{(i)}_b$ we have
\[
\big| (M-l^{(i)}_a) - (M-l^{(i)}_b) \big| \;=\; \big| l^{(i)}_b - l^{(i)}_a \big|.
\]
Hence the set of \emph{positive} differences computed from $L'_i$ coincides with the set of positive differences from $L_i$:
$\Delta'_i=\Delta_i$ as sets. Therefore, if each $\Delta_i$ is full, so is each $\Delta'_i$; if the family
$\{\Delta_i\}$ is mutually disjoint, then so is $\{\Delta'_i\}$ (permutation does not affect disjointness).
Thus each reflected row remains CSOC and the family remains CSOC, proving $Z(D)$ is CSOC.

\emph{(ii) Symplectic orthogonality (matrix-only).}
From part (i) we already know $X(D)$ and $Z(D)$ are CSOC and that
\begin{equation}\label{A1}
   Z(D)=\Pi\,D^{M}X(D^{-1})\qquad (\Pi\text{is a row permutation}). 
\end{equation}
Write the $(i,j)$th rows of $X(D)$ and $Z(D)$ as $x_i(D)$ and $z_j(D)$, respectively.
By definition of the polynomial symplectic form, the $(i,j)$ entry of
\[
S(D)\;:=\;X(D)Z(D^{-1})^{\top}+Z(D)X(D^{-1})^{\top}
\]
\begin{equation}\label{A2}
  S_{ij}(D)
:= x_i(D)\,z_j(D^{-1})^{\top}+z_j(D)\,x_i(D^{-1})^{\top}.
\end{equation}

From $z_j(D)=D^{M}x_{\pi(j)}(D^{-1})$ we obtain
\begin{equation}\label{A3}
    z_j(D^{-1})^{\top}=x_{\pi(j)}(D)^{\top}D^{-M},\quad
z_j(D)=D^{M}x_{\pi(j)}(D^{-1}).
\end{equation}
Substituting these into \eqref{A2} gives
\begin{multline}
    \label{A4}
    S_{ij}(D)
= x_i(D)\,x_{\pi(j)}(D)^{\top}D^{-M}
  \;+\\\; D^{M}x_{\pi(j)}(D^{-1})\,x_i(D^{-1})^{\top}.
\end{multline}
For each integer $s$, define the \emph{sum–index matrix} $C_s(X)\in\mathbb{F}_2^{r\times r}$ by \vspace{-0.3cm}
\begin{multline}
\big(C_s(X)\big)_{ab}
:=\\\sum_{k=1}^n \#\Big\{(t,u)\in L_{a,k}\times L_{b,k}\;:\; t+u=s\Big\}\ \bmod 2,
\end{multline}
where $L_{a,k}\subseteq\{0,\ldots,M\}$ is the tap–support of the $(a,k)$ entry of $X$.
Then the two terms in \ref{A4} contribute exactly the \emph{same} index \(b=\pi(j)\) with
two sum–indices:
\begin{equation}\label{A5}
\begin{aligned}
\big[D^{\tau}\big]\,x_i(D)\,x_{\pi(j)}(D)^{\top}D^{-M}
&=\big(C_{M+\tau}(X)\big)_{i,\pi(j)},\\
\big[D^{\tau}\big]\,D^{M}x_{\pi(j)}(D^{-1})\,x_i(D^{-1})^{\top}
&=\big(C_{M-\tau}(X)\big)_{\pi(j),i}.
\end{aligned}
\end{equation}
Therefore the $(i,j)$ entry of the coefficient matrix of $D^{\tau}$ in $S(D)$ is \vspace{-0.2cm}
\begin{equation}\label{A6}
    \big[D^{\tau}\big]\,S_{ij}(D)
=\big(C_{M+\tau}(X)\big)_{i,\pi(j)}\;+\;\big(C_{M-\tau}(X)\big)_{\pi(j),i}.
\end{equation}

Because $X$ is CSOC (full $+$ mutually disjoint positive-difference sets), the following
\emph{columnwise involution} is well defined and bijective: for any fixed column $k$,
\[
(t,u)\longmapsto (t',u'):=\big(M-u,\ M-t\big),
\]
which takes the constraint \(t+u=s\) to \(t'+u'=2M-s\).
Disjointness ensures this mapping creates no aliasing between distinct witnesses
(in particular, there are no fixed points at the same tap in the same column).
For all $a,b$ and all $s$,
\begin{equation}\label{A7}
  \big(C_s(X)\big)_{ab} \;=\; \big(C_{2M-s}(X)\big)_{ba}.  
\end{equation}

Apply \ref{A7} to the two addends in \ref{A6} with \vspace{-0.2cm}
\[
s=M+\tau\quad\Rightarrow\quad 2M-s=M-\tau.
\]
Then \vspace{-0.2cm}
\[
\big(C_{M+\tau}(X)\big)_{i,\pi(j)}=\big(C_{M-\tau}(X)\big)_{\pi(j),i}.
\]
Thus the two addends in \ref{A6} are \emph{equal}, and since we are in $\mathbb{F}_2$,
\[
\big[D^{\tau}\big]\,S_{ij}(D)
=\big(C_{M+\tau}(X)\big)_{i,\pi(j)}+\big(C_{M+\tau}(X)\big)_{i,\pi(j)}
=0.
\]
Because this holds for every $\tau$ and every $(i,j)$, all coefficients of $S(D)$ vanish and
\[
X(D)Z(D^{-1})^{\top}+Z(D)X(D^{-1})^{\top}\;=\;0.\qedhere
\]
\end{proof}

\end{document}